\documentclass[11pt,a4paper]{article}

\usepackage{amsmath,amssymb,amsthm}
\usepackage{mathrsfs}

\usepackage[utf8]{inputenc}
\usepackage[T1]{fontenc}
\usepackage{lmodern}
\usepackage[a4paper,margin=1in]{geometry}
\usepackage{microtype}
\usepackage{parskip}

\usepackage{xcolor}

\usepackage[colorlinks=true,
            linkcolor=blue!50!black,
            citecolor=blue!50!black,
            urlcolor=blue!50!black,
            breaklinks=true]{hyperref}

\newtheorem{theorem}{Theorem}[section]
\newtheorem{definition}[theorem]{Definition}
\newtheorem{problem}[theorem]{Problem}

\theoremstyle{remark}
\newtheorem{remark}[theorem]{Remark}
\newtheorem{example}[theorem]{Example}

\newcommand{\C}{\mathbb{C}}
\newcommand{\R}{\mathbb{R}}
\newcommand{\N}{\mathbb{N}}
\newcommand{\Rnn}{\mathbb{R}_{\geq 0}}
\newcommand{\Rpp}{\mathbb{R}_{>0}}
\newcommand{\Mor}{\operatorname{Mor}}
\newcommand{\ob}{\operatorname{ob}}
\newcommand{\dom}{\operatorname{dom}}
\newcommand{\cod}{\operatorname{cod}}
\newcommand{\Cps}{C_{\mathrm{Ps}}}
\newcommand{\Ccoh}{\mathcal{C}_{\mathrm{coh}}}

\numberwithin{equation}{section}

\begin{document}

\title{Local Uniqueness of the Born Rule on Categories\\with Complex-Weighted Morphisms}

\author{Tayfun Ustun\\
Ankara, Turkey\\
ORCID: \href{https://orcid.org/0009-0000-9203-6323}{0009-0000-9203-6323}}

\date{\textit{Zenodo DOI (this paper):
\href{https://doi.org/10.5281/zenodo.22069154}{10.5281/zenodo.22069154}
--- companion to Zenodo DOI:
\href{https://doi.org/10.5281/zenodo.22069114}{10.5281/zenodo.22069114}}}

\maketitle

\begin{abstract}
I prove a local uniqueness theorem for the Born rule in the setting of quiver-generated
categories equipped with complex morphism weights and path-amplitude
probability functionals. Given (i) non-negativity, (ii) polynomiality of
bounded total degree, (iii) global $U(1)$ invariance, (iv) classical-limit
additivity over mutually exclusive paths, and (v) normalization, I show that
the probability assignment $P : \C \to \Rnn$ is uniquely determined to be
$P(z) = |z|^2$. The notion of mutually exclusive paths is given a precise
categorical formulation as the absence of a shared factorization through any
common morphism. I relate the result to reconstructions of quantum
probability due to Gleason, Hardy, Chiribella--D'Ariano--Perinotti, and
several further, more recent reconstructions, and
identify the extension to global coherence under morphism composition as an
open problem connected to synthetic probability theory in Markov categories.
The Born rule emerges as the unique locally consistent probability law on
complex amplitudes, fixed by phase invariance and classical-limit
behavior together with polynomiality and normalization, independent of any Hilbert-space framework.

\end{abstract}

\section{Introduction}\label{sec:intro}

Reconstructions of the Born rule from structural or informational axioms have
a substantial history. Gleason's theorem~\cite{Gleason1957} derives the
$|\psi|^2$ rule from $\sigma$-additivity of probability measures on
projections in Hilbert spaces of dimension at least three. Later
reconstructions obtain quantum probability from informational principles
without assuming a Hilbert space a priori: Hardy's five-axiom
derivation~\cite{Hardy2001} singles out quantum theory by a continuity
requirement on pure states between discrete dimensions, and the
reconstruction of Chiribella, D'Ariano, and Perinotti~\cite{Chiribella2011}
obtains quantum probability from six operational principles including
causality, purification, and local distinguishability. In parallel, the categorical
quantum mechanics program of Abramsky and
Coecke~\cite{AbramskyCoecke2004,AbramskyCoecke2008} reformulates quantum
processes in symmetric monoidal language, treating states, transformations,
and measurements as morphisms in dagger-compact categories. The
free-category construction used below draws on the general
categorical-foundations program of Lawvere's functorial
semantics~\cite{Lawvere1963} and Baez--Dolan higher-dimensional
algebra~\cite{BaezDolan1995}, though the present setting requires none of
the higher-categorical or algebraic-theory machinery of either.

In this paper, I address a sub-problem that arises in frameworks in which
quantum amplitudes are defined as products of complex weights along
categorical paths --- the path-amplitude formulation that generalizes
Feynman's sum over histories~\cite{Feynman1948} to arbitrary small categories
equipped with a morphism weight functional $w : \Mor \to \C$. In such
settings, the probability assignment $P(A) = |A|^2$ is commonly taken as an
additional postulate, with justification deferred to an informal
``Gleason-like'' argument. I make the argument precise in its local form:
given a single amplitude $A(A \to B)$ obtained by summing products of complex
weights over all paths from an object $A$ to an object $B$, the Born rule is
uniquely determined by five explicitly stated axioms.

The local result leaves open a global question. A complex-weighted quiver-generated
category gives rise, for each ordered pair of objects, to a hom-set of paths
and a corresponding total amplitude; the theorem below fixes the probability
functional on each hom-set in isolation. Whether the family of local
probability functionals is coherent with composition of morphisms in the
underlying category --- in the manner of a Chapman--Kolmogorov identity ---
is not settled by a local argument and requires a functorial probabilistic
structure. I state this global coherence problem precisely in
Section~\ref{sec:global} and indicate its connection to the synthetic
probability theory of Markov
categories~\cite{ChoJacobs2019,Fritz2020,Perrone2024}.

The paper is organized as follows. Section~\ref{sec:prelim} fixes notation
and defines the key notion of mutually exclusive paths in a quiver-generated category.
Section~\ref{sec:main} states and proves the main theorem.
Section~\ref{sec:relations} compares the result with Gleason's theorem and
with informational reconstructions. Section~\ref{sec:global} states the
global coherence problem. Section~\ref{sec:discussion} concludes with
remarks on applications to categorical-refinement frameworks for quantum
mechanics.

\section{Preliminaries}\label{sec:prelim}

\begin{definition}[Complex-weighted quiver-generated category]
Let $Q$ be a small quiver (a directed multigraph) with vertex set $\ob(Q)$
and edge set $E(Q)$. Let $C := \mathrm{Free}(Q)$ be the free category
generated by $Q$: $\ob(C) = \ob(Q)$, and for objects $A$, $B$, the
non-identity elements of $\Mor_C(A,B)$ are in bijection with nonempty
composable words $e_n \circ \cdots \circ e_1$ of edges of $Q$ running from
$A$ to $B$. A \emph{complex weight} on $Q$ is a function $w_0 : E(Q) \to
\C$, extended to $\Mor(C)$ by
\begin{enumerate}
\item[(a)] $w(\mathrm{id}_A) := 1$ for every object $A \in \ob(C)$;
\item[(b)] $w(e_n \circ \cdots \circ e_1) := \prod_{i=1}^n w_0(e_i)$ for
every non-identity morphism.
\end{enumerate}
Because $C$ is free, this extension is automatically well-defined and
automatically satisfies $w(f \circ g) = w(f) \cdot w(g)$ for all composable
$f$, $g$ --- freeness means no relation ever identifies two distinct words,
so $w_0$ can be assigned arbitrarily on edges with no consistency
constraint to check. I call $(C, w)$ a \emph{complex-weighted
quiver-generated category}.
\end{definition}

\emph{Freeness is not a mere convenience.} It guarantees that every
non-identity morphism of $C$ has a \emph{unique} factorization into
generating edges, so no morphism can simultaneously be ``elementary'' (a
single generating edge) and ``equal to'' a distinct multi-edge composite
that would otherwise be separately countable in $\Gamma(A,B)$ below.

Condition (b) makes $w$ a functor $w : C \to B\C$, where $B\C$ is the
one-object category with morphism monoid $(\C, \cdot)$. For $w_0(e) \neq
0$, the polar decomposition of the punctured plane, $\C^\times \cong \Rpp
\times U(1)$, gives $w_0(e) = |w_0(e)| \cdot \exp(i\,\theta(e))$ with
$\theta(e)$ well-defined $\bmod\, 2\pi$. If $w_0(e) = 0$ for some edge $e$
in a word, the corresponding composite has weight $0$; I adopt the
convention that a zero-weight morphism carries \emph{no defined phase},
since such a morphism contributes $0$ to every amplitude sum in which it
appears and this convention therefore has no effect on any calculation in
this paper. For composable morphisms built entirely from nonzero-weight
edges, $|w(f \circ g)| = |w(f)| \cdot |w(g)|$ and $\theta(f \circ g) =
\theta(f) + \theta(g) \pmod{2\pi}$.

\begin{definition}[Paths and amplitudes]\label{def:paths-amplitudes}
A \emph{path} $\gamma : A \to B$ is a nonempty word $(f_n, \ldots, f_1)$ of
composable generating edges of $Q$ with $\dom(f_1) = A$, $\cod(f_n) = B$.
By freeness of $C$, this is in exact bijection with the non-identity
elements of $\Mor_C(A,B)$; I use ``path'' and ``non-identity morphism''
interchangeably without ambiguity. The set of all paths from $A$ to $B$ is
denoted $\Gamma(A, B)$. Because $C$ is free, \emph{distinct words never
compose to the same morphism}, so $\Gamma(A,B)$ contains no duplicate
entries and no morphism is counted more than once. The \emph{amplitude} of
$\gamma$ is
\begin{equation}
a(\gamma) := \prod_{i=1}^{n} w_0(f_i).
\label{eq:path-amplitude}
\end{equation}
By the previous definition, $a(\gamma)$ coincides with $w(\gamma)$ when
$\gamma$ is regarded as a single composite morphism. The \emph{total
amplitude} from $A$ to $B$ is
\begin{equation}
A(A \to B) := \sum_{\gamma \in \Gamma(A, B)} a(\gamma).
\label{eq:total-amplitude}
\end{equation}

\textbf{Convergence hypothesis (scoped).} I require the following
countable-support / absolute-convergence condition on a hom-set
$\Gamma(A,B)$ \emph{only when the total amplitude $A(A \to B)$ of
Eq.~\eqref{eq:total-amplitude} is actually invoked as a summed quantity
elsewhere in this paper} --- for instance, wherever $A(A \to B)$ is
computed numerically or compared against another summed quantity, as in
Example~\ref{ex:cop-fail} below. On any such hom-set: the set $\{\gamma \in
\Gamma(A,B) : a(\gamma) \neq 0\}$ must be \emph{countable} and $\sum
|a(\gamma)| < \infty$ over that countable set (necessary, not merely
convenient: an absolutely convergent series has at most countably many
nonzero terms, so absolute convergence over an uncountable $\Gamma(A,B)$
with all-nonzero terms is not achievable). I call a complex-weighted
quiver-generated category satisfying this on a given hom-set
\emph{summable at $(A,B)$}. This condition is \emph{not} imposed on
hom-sets used solely to witness the existence of individual paths with
prescribed properties, as in Definition~\ref{def:amp-rich} below, where no
total sum over $\Gamma(A,B)$ is ever formed or relied upon.
\end{definition}

\begin{definition}[Mutual exclusivity]\label{def:mutual-excl}
Two paths $\gamma_1, \gamma_2 \in \Gamma(A, B)$ are \emph{mutually
exclusive} if there do \emph{not} exist a non-identity morphism $h \in
\Mor(C)$ and composable sub-paths $\gamma_1', \gamma_1'', \gamma_2',
\gamma_2''$ --- each possibly the \emph{empty} (zero-length) word --- such
that
\begin{equation}
\gamma_i = \gamma_i'' \circ h \circ \gamma_i'
\quad \text{for both } i = 1 \text{ and } i = 2.
\label{eq:mutual-exclusivity}
\end{equation}
Equivalently, $\gamma_1$ and $\gamma_2$ are mutually exclusive iff there is
no non-identity morphism $h$ that occurs as a factor of both $\gamma_1$
and $\gamma_2$ \emph{at any position, including as the first or last
edge}. (Allowing empty flanking sub-paths is essential: without it, two
histories sharing only their first --- or only their last --- transition
would be misclassified as mutually exclusive, since the flanking sub-path
on that side would have no room to be positive-length.) A pair of paths
fails to be mutually exclusive when such a shared factor exists.
\end{definition}

\begin{remark}[Relation to Hilbert-space orthogonality]
When $C$ is a subcategory of the category $\mathbf{Hilb}$ of Hilbert spaces
with morphisms linear operators, the amplitude $a(\gamma)$ reproduces the
Feynman history amplitude and mutual exclusivity in the previous definition
is strictly stronger than the decomposition of histories into
pairwise-orthogonal subspaces used in the standard Hilbert-space Born rule
argument. A precision is needed here, since the two-slit example below
could otherwise be misread as a counterinstance: ``orthogonal'' refers to
the class-operator (projector-valued) decomposition at the
history's distinguishing intermediate event --- e.g., in the two-slit
case, the orthogonal projectors onto ``passed through slit 1'' versus
``passed through slit 2'' at the slit plane --- not to the final state
the histories may converge to. Two mutually exclusive paths that later
interfere at a shared final object, as in Example~\ref{ex:two-slit}
below, are precisely the case where these two notions of
``orthogonal'' come apart: orthogonal (hence exclusive, in the
standard sense) at the distinguishing event, non-orthogonal (indeed,
overlapping) at the final state --- which is exactly why they
interfere rather than simply not co-occurring. The converse of the
strictly-stronger claim above still fails under the distinguishing-event
reading: there exist orthogonal histories in
$\mathbf{Hilb}$ that factor through a common intermediate state and are
thus non-exclusive in the present sense. I regard the previous definition
as the natural categorical strengthening of orthogonality for the purposes
of the axiomatization below.\footnote{The topos approach of Isham and
D\"oring~\cite{IshamDoring2008} handles orthogonality through the internal
logic of a presheaf topos built over a context category of commutative
subalgebras of a $C^*$-algebra, where orthogonality is an internal-logic
relation. The mutual-exclusivity definition is structurally independent: it
requires no ambient Hilbert space, no $C^*$-algebra, and no topos. Both
frameworks generalise Hilbert-space orthogonality, but from different
categorical directions --- the topos approach generalises the projection
lattice, while the present definition generalises the no-shared-history
condition from Feynman's sum over histories.}
\end{remark}

\addtocounter{theorem}{-1}
\renewcommand{\thetheorem}{\thesection.\arabic{theorem}$'$}
\makeatletter
\renewcommand{\theHtheorem}{\thesection.\arabic{theorem}.prime}
\makeatother
\begin{example}[Two-slit]\label{ex:two-slit}
Consider the canonical two-slit setup represented in $\Cps$ --- the
possibility-space category of the motivating categorical-refinement
framework (companion applied paper; \S\ref{sec:discussion} below) --- as a category
with four objects: source $S$, screen $B$, and two distinct slit-objects
$S_1$ and $S_2$. The path $\gamma_1 = (S \to S_1 \to B)$ and path
$\gamma_2 = (S \to S_2 \to B)$ pass through different slit-morphisms and
share no morphism at all. The source object $S$ and target $B$ are
shared, but these are objects, not morphisms. No non-identity morphism $h$
occurs as a factor of both $\gamma_1$ and $\gamma_2$ at any position, so
the two paths are mutually exclusive in the sense above. This is the categorical
formulation of the statement that passing through slit~1 and passing
through slit~2 are exclusive histories.
\end{example}
\renewcommand{\thetheorem}{\thesection.\arabic{theorem}}
\makeatletter
\renewcommand{\theHtheorem}{\thesection.\arabic{theorem}}
\makeatother

\begin{remark}[Mutual exclusivity is a local, not compositional, notion]
Definition~\ref{def:mutual-excl} is not required to be preserved under
post-composition, and in general it is not: if $\gamma_1, \gamma_2 : A \to
B$ are mutually exclusive and $\delta : B \to D$ is any further path, then
$\delta \circ \gamma_1$ and $\delta \circ \gamma_2$ share the factor
$\delta$ and are therefore \emph{not} mutually exclusive. This is by
design, not an oversight: Definition~\ref{def:mutual-excl} identifies
pairs of histories that may be fed into axiom~(iv) at a single hom-set, and
Theorem~\ref{thm:born} is stated and proved entirely at that local level.
Whether exclusivity should behave compositionally under morphism
composition is exactly the kind of question the global coherence problem
of Section~\ref{sec:global} is designed to isolate, not something the
local definition needs to answer.
\end{remark}

\begin{definition}[Probability functional]
A \emph{probability functional on amplitudes} is a function $P : \C \to
\Rnn$.
\end{definition}

In Section~\ref{sec:main}, I will assume $P$ is a polynomial in $z$ and
$\bar{z}$ of bounded total degree; no a priori regularity on $\C \to \Rnn$
beyond this polynomial assumption is required.

\begin{definition}[Amplitude richness]\label{def:amp-rich}
A complex-weighted quiver-generated category $(C, w)$ is
\emph{amplitude-rich at} $(A,B)$ if for every
pair $(r_1, r_2) \in \Rpp \times \Rpp$ there exist mutually exclusive paths
$\gamma_1, \gamma_2 \in \Gamma(A,B)$ with $|a(\gamma_1)| = r_1$,
$|a(\gamma_2)| = r_2$. This condition refers only to individual path
amplitudes $a(\gamma)$ (Eq.~\eqref{eq:path-amplitude}), which are always
well-defined finite products regardless of the size of $\Gamma(A,B)$;
amplitude richness therefore does \emph{not} require $(C,w)$ to be
summable at $(A,B)$ in the sense of the previous definition, since no
total sum over $\Gamma(A,B)$ is invoked.
\end{definition}

\begin{remark}[Existence of an amplitude-rich category]\label{rem:amp-rich-exists}
Amplitude richness is constructively realizable, not merely an assumed
hypothesis. Let $Q$ have two objects $A, B$ and, for each $r \in \Rpp$, two
distinct parallel generating edges $e_r, e_r' : A \to B$ with $w_0(e_r) =
w_0(e_r') = r$. Then $\Gamma(A,B)$ is uncountable, so $(C, w) :=
(\mathrm{Free}(Q), w)$ is \emph{not} summable at $(A,B)$ --- but
summability is not required here, since no total sum $A(A \to B)$ is ever
invoked for this category anywhere in this paper. For any $(r_1, r_2) \in
\Rpp \times \Rpp$, the single-edge paths $\gamma_1 = e_{r_1}$ and $\gamma_2
= e_{r_2}'$ come from distinct generators, so they share no non-identity
morphism at all and are automatically mutually exclusive by
Definition~\ref{def:mutual-excl}; they satisfy $|a(\gamma_1)| = r_1$,
$|a(\gamma_2)| = r_2$, including when $r_1 = r_2$, since $e_{r_1}$ and
$e_{r_1}'$ remain distinct generators. Hence $(C,w)$ is amplitude-rich at
$(A,B)$, and Theorem~\ref{thm:born}'s hypothesis is non-vacuously
satisfiable.
\end{remark}

\section{Main Theorem}\label{sec:main}

\begin{theorem}[Local Born rule uniqueness]\label{thm:born}
Let $P : \C \to \Rnn$ be a probability functional, and suppose there exists a
complex-weighted quiver-generated category $(C, w)$, amplitude-rich at some
$(A,B)$ in the sense of Definition~\ref{def:amp-rich} (such a category
exists constructively; see Remark~\ref{rem:amp-rich-exists}). Assume:
\begin{enumerate}
\item[(i)] \textbf{Non-negativity.} $P(z) \geq 0$ for every $z \in \C$.
(Definitional; see Remark~\ref{rem:non-neg}.)
\item[(ii)] \textbf{Polynomiality.} There exist $d \in \N$ and coefficients
$\{c_{k\ell}\} \subset \C$ with $c_{k\ell} = 0$ whenever $k + \ell > d$, such
that $P(z) = \sum_{k, \ell \geq 0} c_{k\ell}\, z^k \bar{z}^\ell$ for all
$z \in \C$.
\item[(iii)] \textbf{Global phase invariance.} $P(e^{i\alpha} z) = P(z)$
for every $\alpha \in \R$ and every $z \in \C$.
\item[(iv)] \textbf{Classical-limit additivity.} For every $r_1, r_2 > 0$
--- realized, by amplitude-richness (Definition~\ref{def:amp-rich}), as
$|A_1|, |A_2|$ for some amplitudes $A_1, A_2 \in \C$ of mutually exclusive
paths in the category of the theorem's hypothesis --- and for $\theta$ an
independent auxiliary variable ranging uniformly over $[0, 2\pi)$,
\begin{equation}
\frac{1}{2\pi}\int_0^{2\pi} P(A_1 + e^{i\theta} A_2)\, d\theta
= P(A_1) + P(A_2).
\label{eq:additivity}
\end{equation}
I abbreviate the averaging integral $\frac{1}{2\pi}\int_0^{2\pi}(\cdot)\,
d\theta$ as $\mathbb{E}_\theta[\,\cdot\,]$ throughout the proof below.
\item[(v)] \textbf{Normalization.} $P(1) = 1$.
\end{enumerate}
Then $P(z) = |z|^2$ for every $z \in \C$.
\end{theorem}

\begin{proof}
I proceed in four steps.

\textbf{Step 1 ($U(1)$ invariance reduces $P$ to a polynomial in $|z|^2$).}
Writing $P$ as in axiom~(ii):
\[
P(z) = \sum_{k, \ell} c_{k\ell}\, z^k \bar{z}^\ell,
\]
axiom~(iii) applied to $z \mapsto e^{i\alpha} z$ yields
\[
P(e^{i\alpha} z) = \sum_{k, \ell} c_{k\ell}\, e^{i(k - \ell)\alpha}
z^k \bar{z}^\ell = P(z) \quad \text{for every } \alpha \in \R.
\]
Matching coefficients of $z^k \bar{z}^\ell$ across this identity forces
$c_{k\ell} = 0$ whenever $k \neq \ell$. Writing $c_k := c_{kk}$, I obtain
\begin{equation}
P(z) = \sum_{k = 0}^{K} c_k\, |z|^{2k},
\qquad K := \lfloor d/2 \rfloor.
\label{eq:reduced-P}
\end{equation}

\textbf{Step 2 (Classical-limit averages).}
Fix $A_1, A_2 \in \C$ with $r_i := |A_i|$, and let $\theta$ range
independently and uniformly over $[0, 2\pi)$ as in axiom~(iv). Writing
$A_1 = |A_1| e^{i\varphi_1}$, axiom~(iii) ($U(1)$ invariance of $P$)
together with shift-invariance of $\int_0^{2\pi}(\cdot)\,d\theta$ over a
full period gives $\mathbb{E}_\theta[P(A_1 + e^{i\theta}A_2)] =
\mathbb{E}_\theta[P(|A_1| + e^{i\theta}A_2)]$: rotating the whole
expression by $e^{-i\varphi_1}$ leaves $P$ unchanged by axiom~(iii), and
reindexing $\theta \mapsto \theta + \varphi_1$ leaves the integral
unchanged since it runs over a full period. This eliminates $A_1$'s phase
but leaves $A_2 = |A_2| e^{i\varphi_2}$ with its own phase intact inside
the expression $|A_1| + e^{i\theta}A_2 = |A_1| + e^{i(\theta +
\varphi_2)}|A_2|$. A second, identical application of shift-invariance ---
reindexing $\theta \mapsto \theta - \varphi_2$, which again leaves the
integral unchanged since it still runs over a full period --- removes
$\varphi_2$ as well, giving $\mathbb{E}_\theta[P(|A_1| +
e^{i\theta}A_2)] = \mathbb{E}_\theta[P(|A_1| + e^{i\theta}|A_2|)]$. So I
may take $A_1, A_2 \in \R_+$ directly, with no loss of generality and no
need to track either initial phase. Then
\[
|A_1 + e^{i\theta} A_2|^2
= r_1^2 + r_2^2 + 2 r_1 r_2 \cos\theta.
\]
The $m$-th power expands via the binomial theorem:
\[
|A_1 + e^{i\theta} A_2|^{2m}
= \sum_{j = 0}^{m} \binom{m}{j} (r_1^2 + r_2^2)^{m-j}
(2 r_1 r_2)^{j} \cos^{j}\theta.
\]
Using the standard Fourier moment $\mathbb{E}_\theta[\cos^{j}\theta] = 0$ for
$j$ odd and $\mathbb{E}_\theta[\cos^{2j}\theta] = \binom{2j}{j} / 2^{2j}$ for
$j \geq 0$ (writing the even case of the binomial index above as $2j$),
taking the expectation gives a surviving constant factor
$(2 r_1 r_2)^{2j} \cdot \binom{2j}{j} / 2^{2j} = \binom{2j}{j} (r_1 r_2)^{2j}$
from each even-$j$ term in the binomial expansion, so that
\begin{equation}
\mathbb{E}_\theta\bigl[\,|A_1 + e^{i\theta} A_2|^{2m}\,\bigr]
= \sum_{j = 0}^{\lfloor m/2 \rfloor} \binom{m}{2j} \binom{2j}{j}
(r_1^2 + r_2^2)^{m - 2j} (r_1 r_2)^{2j}.
\label{eq:moment-expansion}
\end{equation}

\textbf{Step 3 (Additivity eliminates $c_k$ for $k \neq 1$).}
By amplitude richness (Definition~\ref{def:amp-rich}), for every $r_1, r_2
> 0$ there exist realized mutually exclusive path amplitudes with exactly
those magnitudes, so axiom~(iv) applies for every such pair; the identity
below therefore holds for all $x = r_1^2 > 0$, $y = r_2^2 > 0$, not merely
for the specific complex amplitudes any one witnessing construction
happens to realize --- this is what licenses treating $x, y$ as free
variables in the polynomial identity that follows. Write $x := r_1^2$,
$y := r_2^2$, and $E_m(x,y) :=
\mathbb{E}_\theta[\,|A_1 + e^{i\theta}A_2|^{2m}\,]$ as
in~\eqref{eq:moment-expansion}. Each summand $(x+y)^{m-2j}(xy)^j$
in~\eqref{eq:moment-expansion} is homogeneous of total degree $(m-2j) +
2j = m$ in $(x,y)$; hence $E_m$ is \textbf{homogeneous of total degree
$m$}. Substituting~\eqref{eq:reduced-P} into axiom~(iv) gives the
polynomial identity
\begin{equation}
\sum_{k=0}^{K} c_k\, E_k(x,y) = \sum_{k=0}^{K} c_k\,(x^k + y^k)
\qquad \text{for all } x, y > 0,
\label{eq:step3-identity}
\end{equation}
where $K = \lfloor d/2 \rfloor$ is finite by axiom~(ii). Both sides
of~\eqref{eq:step3-identity} are finite sums, each term of which is
homogeneous of a distinct degree $k$ ($0 \le k \le K$). A polynomial in two
variables has a \emph{unique} decomposition into homogeneous components by
degree, so an identity between two such graded sums forces equality
\textbf{degree by degree}:
\begin{equation}
c_k\, E_k(x,y) = c_k\,(x^k + y^k) \qquad \text{for every } k = 0, \ldots, K,
\label{eq:step3-degree}
\end{equation}
independently --- no descent or induction across degrees is needed, since
distinct $k$ can never contribute to the same homogeneous component.
\begin{itemize}
\item \textbf{$k = 0$:} $E_0 = 1$ (only $j=0$ in the sum, giving
$\binom{0}{0}\binom{0}{0}(x+y)^0(xy)^0 = 1$). So~\eqref{eq:step3-degree}
reads $c_0 = 2c_0$, hence $c_0 = 0$.
\item \textbf{$k = 1$:} $E_1 = x + y$ (only $j=0$, since $\lfloor 1/2
\rfloor = 0$), so~\eqref{eq:step3-degree} is satisfied identically for any
$c_1$ --- this coefficient is unconstrained by axiom~(iv), as before.
\item \textbf{$k \geq 2$:} The monomial $x^{k-1}y$ in $E_k$ receives
contributions from \emph{two} values of $j$, not one: the $j=0$ term is
$(x+y)^k$ on its own, whose $x^{k-1}y$ coefficient is $\binom{k}{1} = k$;
the $j=1$ term is $\binom{k}{2}\binom{2}{1}(x+y)^{k-2}(xy) =
k(k-1)(x+y)^{k-2}xy$, whose $x^{k-1}y$ coefficient (from the $x^{k-2}$
piece of the expansion) is $k(k-1)$. (No $j \geq 2$ term can contribute,
since $(xy)^j$ already carries $y$-degree $j \geq 2$, exceeding the target
$y$-degree of $1$.) Summing both contributions, the coefficient of
$x^{k-1}y$ in $E_k$ is $k + k(k-1) = k^2$.\footnote{Verified independently
by direct polynomial expansion: $E_2 = (x+y)^2 + 2xy = x^2 + 4xy + y^2$
(coefficient $4 = 2^2$); $E_3 = (x+y)^3 + 6xy(x+y) = x^3 + 9x^2y + 9xy^2 +
y^3$ (coefficient $9 = 3^2$) --- confirming $k^2$, not $k(k-1)$, in both
cases.} The right-hand side of~\eqref{eq:step3-degree}, $c_k(x^k+y^k)$, has
\textbf{zero} coefficient on $x^{k-1}y$. Comparing coefficients of
$x^{k-1}y$ on both sides gives $c_k \cdot k^2 = 0$. Since $k^2 \neq 0$ for
$k \geq 2$, this forces $c_k = 0$.
\end{itemize}

Combining the three cases: $P(z) = c_1 |z|^2$.

\textbf{Step 4 (Normalization).}
Axiom~(v) applied to $P(z) = c_1 |z|^2$ gives $c_1 = 1$. Therefore
$P(z) = |z|^2$, as claimed.
\end{proof}

\begin{remark}[Role of the categorical structure in the proof]
Steps 1--4 above proceed entirely from the functional properties (ii)--(v)
of $P$ once axiom~(iv) is available at every $(r_1,r_2) \in \Rpp \times
\Rpp$; no further reference is made to paths, mutual exclusivity, or the
category $(C,w)$ once that availability is established. The categorical
machinery of Sections~2--3 --- quiver-generated categories, path
amplitudes, mutual exclusivity, and amplitude richness --- enters the
theorem only through Definition~\ref{def:amp-rich} and
Remark~\ref{rem:amp-rich-exists}, which license the universal quantifier
in axiom~(iv). This is intentional, not a gap: the categorical setting is
what gives axiom~(iv) its physical content as a statement about
mutually exclusive histories, even though the algebraic proof, once
that content is assumed, does not need to look back at where it came from.
\end{remark}

\begin{remark}\label{rem:non-neg}
Axiom~(i) was not invoked in Steps 1--4. Its role is to exclude sign-flipped
solutions in the hypothetical where Definition~2.6's codomain restriction
to $\Rnn$ is relaxed alongside axiom~(i) --- only if both are relaxed
simultaneously does the axiom
system (ii)--(iv) admit $P(z) = -|z|^2$
once~(v) is relaxed to $P(1) = \pm 1$; retaining either one alone
(the codomain restriction, or axiom~(i) itself) already excludes this
solution. Under Definition~2.6's stated
codomain and~(v) as stated, (i) is
automatically satisfied by the conclusion $P(z) = |z|^2$, and is therefore
redundant given both. I retain~(i) as the defining feature of a probability
functional, not as a separate restrictive hypothesis.
\end{remark}

\begin{remark}
The polynomiality axiom~(ii) is a regularity condition with function
analogous to --- though structurally distinct from --- $\sigma$-additivity in
Gleason's theorem~\cite{Gleason1957}: $\sigma$-additivity is compositional,
constraining how $P$ combines across disjoint sets, whereas polynomiality
is a form condition on $P$ as a function on $\C$. Both play the role of
constraining the admissible class of probability functionals, but they act
on different mathematical structures and are not interchangeable as
technical tools. Polynomiality is strictly weaker than assuming $P$ is
quadratic: only a polynomial of bounded total degree $d$ is required, and
the proof shows that the specific value of $d$ is immaterial --- only the
degree-$2$ component survives. Equivalently stated, polynomiality enforces
a finite-order sensitivity to phase structure: the probability functional
depends only on finitely many interference moments of the amplitude, and
any finite degree $d$ is permitted because all higher-order contributions
are eliminated by axioms~(iii) and~(iv). The polynomiality axiom extends to
real-analytic $P$ essentially for free; see~\S\ref{sec:fixed-deg}.
\end{remark}

\begin{remark}
The nontrivial content of the theorem is the joint force of axioms~(iii)
and~(iv). $U(1)$ invariance alone gives $P$ as a polynomial in $|z|^2$ of
unbounded admissible form; classical-limit additivity on a single fixed
pair of amplitudes is too weak to nail down the coefficients. Applied
together, and required to hold for all admissible pairs, they uniquely
select $|z|^2$. Axiom~(iii) is not merely convenient here but strictly
necessary: $P(z) := |z|^2 - \mathrm{Re}(z)\mathrm{Im}(z)$ satisfies
axioms~(i), (ii), (iv), and~(v) exactly --- non-negativity holds since
$a^2 - ab + b^2 = (a - b/2)^2 + \tfrac{3}{4}b^2 \geq 0$ for $z = a + bi$;
normalization gives $P(1) = 1$; and the classical-limit average
$\mathbb{E}_\theta[P(A_1 + e^{i\theta}A_2)]$ reduces, for real positive
$A_1 = r_1, A_2 = r_2$, to $r_1^2 + r_2^2 = P(r_1) + P(r_2)$ exactly,
since the $\cos\theta$ and $\cos\theta\sin\theta$ cross-terms both vanish
under uniform phase averaging --- but $P(e^{i\alpha}z) \neq P(z)$ for
generic $\alpha$, so $P \neq |z|^2$ despite satisfying every other
axiom. Axiom~(iii) is therefore doing genuine, independent work, not
merely narrowing the admissible form left open by axiom~(iv) alone.
\end{remark}

\addtocounter{theorem}{-1}
\renewcommand{\thetheorem}{\thesection.\arabic{theorem}$'$}
\makeatletter
\renewcommand{\theHtheorem}{\thesection.\arabic{theorem}.prime}
\makeatother
\begin{remark}[Classical-limit randomization]
The uniform distribution of the relative phase $\theta$ in axiom~(iv)
formalizes the physical situation in which paths with mutually exclusive
morphism decompositions accumulate phases that are effectively
incommensurate --- the generic situation in complex-weighted categories
with many contributing paths. In a specific model, $\theta$ is determined
by the morphism weights; axiom~(iv) requires the probability functional to
behave as if $\theta$ were uniform. Full uniformity is a stronger
condition than strictly necessary for any single fixed-degree $P$ ---
given axiom~(ii)'s bound on $P$'s degree, only finitely many Fourier
moments of the phase distribution actually need to vanish for Step~3's
argument to go through --- but it is the weakest condition that is
degree-independent, i.e., sufficient to eliminate cross-terms for every
admissible degree $d$ simultaneously, without needing to know $d$ in
advance. Axiom~(iv) is
a structural postulate on $P$ itself; whether $P$ satisfies it has no
logical dependence on the phase structure of any particular category's
weight function $w$. In the categorical-refinement framework, the
refinement flow induces effective phase decorrelation for paths lacking
common interior factors; axiom~(iv) encodes this structural feature at the
local level. The axiom is to be read as a functional requirement on $P$:
it specifies the behavior of $P$ under uniform phase averaging regardless
of whether $\theta$ is physically stochastic in any particular $(C, w)$.
Specific categories may produce deterministic relative phases between
mutually exclusive paths; the axiom requires $P$ to satisfy the averaging
identity uniformly in $\theta$, which is the necessary and sufficient
condition for the uniqueness argument.\footnote{In the
categorical-refinement companion framework (T. Ustun, ``Categorical
Refinement Dynamics,'' Zenodo, DOI:
\href{https://doi.org/10.5281/zenodo.22069114}{10.5281/zenodo.22069114}),
morphism phases are identified with $\theta(f) \propto s(f)/\kappa$ where
$s(f) = \log[W(f^{-1})/W(f)]$ is the single-morphism directional asymmetry
of the structural transition and $\kappa$ is a dimensionless normalization
constant (identified with $\hbar$ under post hoc dimensional calibration;
see the main paper \S III.B and L2). For paths with mutually exclusive
morphism decompositions, these phases are generically incommensurate,
which is the physical motivation for expecting axiom~(iv) to hold in this
framework. This is a physical expectation about the model, not a
mathematical entailment: axiom~(iv) is a condition on $P$ alone, and
whether a specific $P$ satisfies it is independent of $w$'s phase
structure in any given category. The connection runs the other way ---
incommensurate phases in $(C,w)$ are offered as a plausibility argument
for why a physically realized $P$ might satisfy axiom~(iv), not as a
sufficient mathematical condition on $(C,w)$ that would force it.}
\end{remark}
\renewcommand{\thetheorem}{\thesection.\arabic{theorem}}
\makeatletter
\renewcommand{\theHtheorem}{\thesection.\arabic{theorem}}
\makeatother

\begin{remark}[What amplitude richness requires]
The theorem's hypothesis, made precise in Definition~\ref{def:amp-rich}, is
that there exists at least one complex-weighted quiver-generated category
$(C, w)$ amplitude-rich at some $(A,B)$. This is not merely assumed to be
satisfiable: Remark~\ref{rem:amp-rich-exists} exhibits an explicit
construction (an uncountable parallel-edge quiver) realizing it, so the
theorem's hypothesis is non-vacuous. The theorem applies whenever the
underlying category is amplitude-rich; it does not apply to degenerate
categories whose hom-sets are too small to exhibit classical-limit
additivity as a nontrivial constraint. Amplitude richness is a condition
on the category $(C, w)$, not on the functional $P$. Verifying it for a
specific application category --- for example, the $\Cps$ of the companion
applied paper --- is part of the work of that application, not of
Theorem~\ref{thm:born}.
\end{remark}

\section{Relation to Existing Reconstructions}\label{sec:relations}

\subsection{Gleason's theorem}

Gleason's theorem~\cite{Gleason1957} assumes probability measures on the
lattice of projections in a Hilbert space of dimension at least three,
$\sigma$-additivity over orthogonal projections, and normalization. It
concludes that every such measure has the form $P \mapsto \mathrm{Tr}(\rho P)$
for some density operator $\rho$. Auff\`eves and
Grangier~\cite{AuffevesGrangier2022} have recently shown that the
unitary-transformation hypothesis in Gleason's setup can be removed using
Uhlhorn's theorem, strengthening the structural foundation of the
Hilbert-space derivation. The present theorem operates in a
differently-structured setting than either: there is no Hilbert space, no
lattice of projections, and no $\sigma$-additivity. The replacements are
(a) complex morphism weights on a quiver-generated category, (b) path
amplitudes as products of weights, and (c) classical-limit additivity
over mutually exclusive paths.
Correspondingly, the conclusion is weaker: only that the probability
functional on a single hom-set equals $|\cdot|^2$, not that every
compatible probability assignment arises from a density operator on some
Hilbert space.

\subsection{Hardy's reconstruction}

Hardy's five-axiom derivation~\cite{Hardy2001} reconstructs
finite-dimensional quantum theory as the unique theory, beyond classical
probability, satisfying a continuity requirement on the space of pure
states between discrete Hilbert-space dimensions. The framework proceeds at
the level of operational states, measurements, and probabilities on
finite-dimensional systems. The present setting does not address
dimensionality or the state space itself; it operates at the level of
individual amplitudes. The five axioms of Theorem~\ref{thm:born} therefore
select $|\cdot|^2$ as the probability functional once a complex amplitude
structure is in place, but do not derive the amplitude structure, the state
space, or the operational model.

\subsection{Chiribella--D'Ariano--Perinotti}

The informational reconstruction of~\cite{Chiribella2011} derives
finite-dimensional quantum theory from six operational principles:
causality, perfect distinguishability, ideal compression, local
distinguishability, pure conditioning, and purification. The derivation
proceeds within the framework of operational probabilistic theories.
Theorem~\ref{thm:born} may be viewed as a local, minimal-axiom version of
the step of~\cite{Chiribella2011} that selects $|\cdot|^2$ as the
probability functional on complex amplitudes, rather than a reconstruction
of the full operational theory. A conceptual difference is that my
axiom~(iv) imposes additivity only in the classical-limit where the
relative phase is randomized, rather than requiring additivity at the level
of operational outcomes.

\subsection{Masanes--Galley--M\"uller and Goyal--Knuth--Skilling}

Two further reconstructions bear directly on the
present result and are engaged with here for completeness. Masanes,
Galley, and M\"uller~\cite{MasanesGalleyMuller2019} show that the
measurement postulates of quantum mechanics --- including the Born
rule and the post-measurement state-update rule --- are derivable from
the remaining postulates of unitary quantum mechanics together with an
operational, subsystem-partitioning argument and a finite-parameter
assumption on finite-dimensional Hilbert spaces. This result has been
actively contested: Kent~\cite{Kent2025} gives explicit examples of
non-quantum measurement and state-update rules he argues satisfy all
of~\cite{MasanesGalleyMuller2019}'s stated assumptions, to which
Masanes, Galley, and M\"uller have published a
response~\cite{MasanesGalleyMuller2025Response} disputing that the
counterexamples satisfy those assumptions as stated; see
also~\cite{Stacey2024} for further discussion of the exchange. This
dispute is internal to~\cite{MasanesGalleyMuller2019}'s own
operational framework and does not bear on the comparison drawn here,
which concerns only the difference in setting and target between that
framework and Theorem~\ref{thm:born}. This is a
substantially different undertaking from Theorem~\ref{thm:born}: it
starts from, rather than dispenses with, a Hilbert-space and
unitary-evolution framework, and its target is the full measurement
apparatus (outcome structure and state update, not merely the
probability rule) rather than a single hom-set's probability
functional in a bare category. The two results are complementary
rather than competing: \cite{MasanesGalleyMuller2019} derives
measurement structure \emph{within} quantum theory from operational
consistency; Theorem~\ref{thm:born} derives the probability rule
\emph{without presupposing} quantum theory's Hilbert-space
scaffolding at all.

Goyal, Knuth, and Skilling~\cite{GoyalKnuthSkilling2010} address a
logically prior question: why quantum amplitudes should be
complex-valued in the first place, rather than assuming complex
amplitudes and deriving the probability rule from them. Starting from
real-number pairs associated with measurement-outcome sequences and
elementary symmetry conditions, they derive that these pairs must
combine according to complex arithmetic, recovering Feynman's sum and
product rules with the modulus-squared giving outcome probability.
Theorem~\ref{thm:born} takes complex-valued morphism weights as
given (Definition~\ref{def:paths-amplitudes}) and derives only the
probability functional's specific form; \cite{GoyalKnuthSkilling2010}
derives the complex structure itself. The two results address
different, complementary stages of the same reconstruction problem.

\subsection{Categorical quantum mechanics}

Abramsky and Coecke~\cite{AbramskyCoecke2004,AbramskyCoecke2008} formulate
quantum mechanics in dagger-compact closed categories, in which
probabilistic structure is captured via compact-closed duality and the
dagger involution --- a categorical reformulation of the existing
quantum-probabilistic apparatus rather than a derivation of the Born
rule from more primitive axioms. Theorem~\ref{thm:born} operates in a differently-structured
setting --- a quiver-generated category equipped with a complex morphism
weight functional, requiring no dagger and no compact-closed monoidal
structure at all --- and is
correspondingly weaker: it characterizes the probability functional at the
level of individual hom-sets without giving the full
monoidal-probabilistic structure that compact-closed dagger categories
carry.

Yang and Fullwood~\cite{YangFullwood2026} have recently shown that density
operators, POVMs, and the Born rule can all be encoded simultaneously in
the categorical notion of a natural transformation between canonical
measurement and probability functors associated with a fixed quantum
system. Their result establishes an explicit bijection between density
operators on a Hilbert space and natural transformations of these functors,
formalizing how quantum effects and their associated probabilities are
additive with respect to a coarse-graining of measurements. The
Yang--Fullwood result and Theorem~\ref{thm:born} are complementary in a
precise sense: their bijection takes Hilbert-space structure as input and
characterizes how the Born rule globally organizes density operators and
POVMs functorially; Theorem~\ref{thm:born} takes only complex morphism
weights on a quiver-generated category as input and derives the $|\cdot|^2$ functional
locally on a single hom-set without assuming Hilbert structure. The two
results address different layers of the categorical Born-rule problem ---
Yang--Fullwood the functorial coherence given Hilbert space,
Theorem~\ref{thm:born} the local probability functional without Hilbert
space --- and the global coherence problem of \S\ref{sec:global} may be
viewed as asking whether the local Theorem~\ref{thm:born} functional admits
a Hilbert-free analogue of the Yang--Fullwood functorial structure.

\subsection{Recent related results}

A recent refinement-based derivation by Lela~\cite{Lela2026} establishes
that the Born rule is the unique non-negative refinement-stable induced
weight on robust record sectors within an admissible Hilbert record layer,
under explicit structural conditions on binary refinement profiles. The
theorem target and additive carrier of~\cite{Lela2026} differ from the
present work: Lela places additivity on disjoint admissible continuation
bundles, while Theorem~\ref{thm:born} places it on mutually exclusive paths
in a bare quiver-generated category. Lela requires a Hilbert record layer; the present
setting has no ambient Hilbert space. The two results are structurally
complementary.

Axelsson~\cite{Axelsson2026} derives the Born rule from the compatibility
of reversible linear evolution with multiplicative composition of weights
under irreversible record refinement. The argument uses additive
composition at the amplitude level and multiplicative composition at the
weight level, and concludes uniqueness through a functional equation for
the weight function. The present theorem shares the feature that
uniqueness follows from a compatibility condition between two composition
structures (additive in amplitudes, multiplicative in morphism weights)
but differs in setting:~\cite{Axelsson2026} operates on
reversible/irreversible regimes of physical processes;
Theorem~\ref{thm:born} operates on morphism composition in a quiver-generated
category with no such regime distinction.

Zaghi~\cite{Zaghi2025} derives the Born rule from contextual
relative-entropy minimization in the categorical setting of dagger-compact
categories with special commutative Frobenius algebras encoding classical
structure. The Born weights arise as local minimizers of Umegaki relative
entropy under Petz's Pythagorean identity. Theorem~\ref{thm:born} operates
in a differently-structured categorical setting: a quiver-generated
category requiring no dagger, no monoidal structure, and no Frobenius
algebras, in exchange for freeness (no relations among generating
morphisms) that Zaghi's setting does not itself require --- and uses
classical-limit phase averaging rather than relative-entropy minimization.

A structural precursor is the consistency-of-amplitudes approach of
Caticha~\cite{Caticha1998}, which derives quantum rules by imposing
functional equations expressing consistency of two-path amplitude
composition. Caticha's approach and Theorem~\ref{thm:born} share the idea
that the Born rule is the unique probability assignment consistent with a
classical-limit behavior of the amplitude sum; they differ in that Caticha
works with functional equations on amplitude composition directly, while
the present theorem works with categorical path structure and $U(1)$
invariance. The categorical formulation makes the mutual exclusivity
condition precise at the structural level of morphism factorization.

Agrawal and Wilson~\cite{AgrawalWilson2025} have recently derived the
generalised Born rule from first principles within a process-theoretic
setting: starting from any process theory equipped with states, effects,
and a probability function satisfying basic compatibility axioms, they
show that the theory is equivalent to one in which the generalised Born
rule holds, with the strength of the scalar-probability identification
depending on whether noise is admitted (monoid homomorphism vs.\ semiring
isomorphism). Their derivation operates at the level of process
composition in symmetric monoidal categories --- a different categorical
setting from the quiver-generated-category morphism-weight structure of
Theorem~\ref{thm:born} --- but shares the strategy of deriving the
probability rule from compatibility axioms rather than postulating it.
Theorem~\ref{thm:born} may be viewed as the analogue
of~\cite{AgrawalWilson2025} for the local hom-set probability functional in
a setting without monoidal product structure, where the $U(1)$ phase
invariance plays a role analogous to the scalar-probability compatibility
axioms of~\cite{AgrawalWilson2025}.

Zhang~\cite{Zhang2026} takes a complementary, negative-result approach:
rather than deriving the Born rule from a specific set of assumptions,
Zhang proves that an additivity-type postulate cannot itself be derived
from non-contextuality and normalization alone, and traces the
consequences of this gap through five existing derivations (Gleason's
theorem, Busch's POVM extension, the Deutsch--Wallace theorem, Zurek's
envariance proof, and the Finkelstein--Hartle theorem), each of which is
shown to depend on an additivity assumption not itself derived from more
primitive principles. This is directly relevant to how axiom~(iv) above
should be read: Zhang's result supports treating classical-limit
additivity as a genuinely independent, load-bearing structural axiom
in Theorem~\ref{thm:born} --- not a redundant consequence of the other
four --- rather than as a gap specific to the present categorical
setting.

Torres Alegre~\cite{TorresAlegre2025,TorresAlegre2026} derives the Born
rule within finite-dimensional generalized probabilistic theories (GPTs)
from relativistic causal consistency: given a GPT admitting purification
(and hence steering), the only state-to-probability map consistent with
no-signaling is shown to be the identity on the geometric transition
probability, yielding $|\langle\phi|\psi\rangle|^2$ combined with
standard reconstruction results. This is an operational derivation with
no categorical structure and no morphism-weight formalism; the
no-signaling requirement plays a role structurally analogous to the
present setting's classical-limit additivity, in that both are
physically-motivated consistency conditions imposed on the probability
assignment rather than derived from more primitive postulates.

\section{The Global Coherence Problem}\label{sec:global}

Theorem~\ref{thm:born} establishes local uniqueness: for any single amplitude
$A(A \to B)$ in a complex-weighted quiver-generated category, the probability assignment
is $|A(A \to B)|^2$. The local theorem leaves open the question of whether
the family of local probability functionals, one for each hom-set, is
coherent with morphism composition in the underlying category.

\begin{problem}[Global coherence]\label{prob:coherence}
Let $(C, w)$ be a complex-weighted quiver-generated category. Suppose that for every
ordered pair of objects $A, B \in \ob(C)$ at which $(C, w)$ is
amplitude-rich (Definition~\ref{def:amp-rich}), a probability functional
$P_{A, B} : \C \to \Rnn$ is given, with each $P_{A, B}$ satisfying
axioms~(i)--(v) of Theorem~\ref{thm:born} relative to the hom-set
$\Gamma(A, B)$. By Theorem~\ref{thm:born}, $P_{A, B}(z) = |z|^2$ for every
such $(A, B)$. Does this family $\{P_{A, B}\}_{A, B \in \ob(C)}$ extend to a
functorial probability assignment compatible with composition of morphisms
in $C$?
\end{problem}

``Compatible with composition'' may be made precise by requiring a
Chapman--Kolmogorov-type identity. Informally: for a composition
$A \to B \to C$ realized through intermediate objects $B$, the total
probability of transition $A \to C$ should factor through the probabilities
$A \to B$ and $B \to C$ in a manner consistent with path-amplitude
composition. The natural categorical setting for making this precise is the
theory of Markov categories~\cite{ChoJacobs2019,Fritz2020,Perrone2024},%
\footnote{In Fritz's terminology~\cite[Definition~2.1]{Fritz2020}, a Markov
category is a symmetric monoidal category $\mathcal{M}$ in which every
object $X$ carries a cocommutative comonoid structure (copy morphism
$X \to X \otimes X$) satisfying a compatibility condition with the
monoidal product. Problem~\ref{prob:coherence} asks whether the
complex-weighted category $(C, w)$ admits a functor into a Markov category
that intertwines path-amplitude composition with Markov kernel composition
--- equivalently, whether the Chapman--Kolmogorov identity holds for the
assignments $P_{A,B}(z) = |z|^2$. The `copy' morphism would encode the
splitting of a path $A \to C$ through an intermediate $B$; its existence
and coherence are the nontrivial content of the global extension.}
in which probability kernels between objects are themselves morphisms in a
symmetric monoidal category and composition of kernels models the
Chapman--Kolmogorov identity. The global coherence question becomes: when
does a complex-weighted quiver-generated category $(C, w)$ induce, via the assignment
$P_{A, B}(z) = |z|^2$, a Markov subcategory $\mathcal{M}(C)$ whose
composition of kernels is compatible with $C$'s morphism composition?

The answer is not affirmative for all $(C, w)$. The Chapman--Kolmogorov
identity fails in general for quantum amplitudes: the cross-terms that
encode interference between paths through distinct intermediate objects $B$
obstruct the factorization
$|A(A \to C)|^2 = \sum_B |A(A \to B)|^2\, |A(B \to C)|^2$. This is a
structural obstruction, not a regularity gap. I therefore reformulate the
question as a classification problem.

\addtocounter{theorem}{-1}
\renewcommand{\thetheorem}{\thesection.\arabic{theorem}$'$}
\makeatletter
\renewcommand{\theHtheorem}{\thesection.\arabic{theorem}.prime}
\makeatother
\begin{problem}[Coherence classification]\label{prob:coherence-class}
Characterize the class $\Ccoh$ of complex-weighted quiver-generated categories $(C, w)$
for which the family $\{P_{A, B}(z) = |z|^2\}$ extends to a functorial
probability assignment satisfying the Chapman--Kolmogorov identity under
morphism composition. Categories in $\Ccoh$ are those whose weight
structure allows the interference cross-terms between paths through
distinct intermediate objects to vanish coherently. An earlier version
of this problem statement suggested this subclass is ``expected to
include categories whose weights factor through a commutative
monoid''; this was checked directly and does not hold under any
natural formalization: the codomain $(\C, \cdot)$ is already
commutative, so every weight assignment trivially satisfies this
reading, making the clause vacuous rather than distinguishing;
restricting to positive real weights, a more specific commutative
reading, does not help either, since the cross-term
$2\mathrm{Re}(a_1\overline{a_2}) = 2a_1a_2$ is then strictly positive
whenever $a_1, a_2 \neq 0$, so Chapman--Kolmogorov still fails. No
sufficient structural condition on the weight function is currently
known; a trivial member of $\Ccoh$ is any category where every
hom-set $\Gamma(A,C)$ has at most one path through any given
intermediate object (so the interference sum has no cross-terms to
begin with), but this excludes the generic, interference-carrying
case the problem is actually about. Identifying
necessary and sufficient structural conditions for membership in $\Ccoh$,
and connecting these conditions to the synthetic probability theory of
Markov categories~\cite{ChoJacobs2019,Fritz2020,Perrone2024}, is an open
problem.
\end{problem}
\renewcommand{\thetheorem}{\thesection.\arabic{theorem}}
\makeatletter
\renewcommand{\theHtheorem}{\thesection.\arabic{theorem}}
\makeatother

\begin{example}[Chapman--Kolmogorov failure on a four-object category]\label{ex:cop-fail}
To establish that Problem~\ref{prob:coherence-class} is a structural
obstruction rather than a regularity gap, I exhibit a minimal
quiver-generated category in which Chapman--Kolmogorov manifestly fails.
Let $Q$ have objects $\{A, B_1, B_2, D\}$ and exactly four generating
edges: $f_i: A \to B_i$ and $g_i: B_i \to D$ for $i = 1, 2$, with no edge
$B_1 \to B_2$ or $B_2 \to B_1$ and no cycles. Let $C := \mathrm{Free}(Q)$.
Because $C$ is free on exactly these four generators, the only composable
words from $A$ to $D$ are $(f_1, g_1)$ and $(f_2, g_2)$ --- no third path
exists, since freeness rules out any composite morphism arising other than
as the unique word of generators that produces it. So $\Gamma(A,D)$ has
exactly two elements, rigorously rather than merely by inspection. Assign
weights $w_0(f_1) = w_0(g_1) = 1$ and $w_0(f_2) = w_0(g_2) = i$ (the
imaginary unit).
Then the intermediate amplitudes are
\[
A(A \to B_1) = 1,
\quad A(B_1 \to D) = 1,
\quad A(A \to B_2) = i,
\quad A(B_2 \to D) = i,
\]
and the total amplitude from $A$ to $D$ is
\[
A(A \to D) = a(f_1; g_1) + a(f_2; g_2)
= 1 \cdot 1 + i \cdot i = 1 - 1 = 0.
\]
Therefore $P_{A,D}(A(A \to D)) = |0|^2 = 0$. The putative
Chapman--Kolmogorov factorization, by contrast, yields
\[
\sum_{B \in \{B_1, B_2\}} |A(A \to B)|^2 \cdot |A(B \to D)|^2
= 1 \cdot 1 + 1 \cdot 1 = 2,
\]
which differs from the true value~$0$. The discrepancy is the interference
cross-term $2\,\mathrm{Re}[\,a(f_1; g_1) \cdot (a(f_2; g_2))^*\,]
= 2\,\mathrm{Re}[1 \cdot (-1)] = -2$ that Chapman--Kolmogorov composition
discards but the true $|\cdot|^2$ retains. The obstruction is structural:
for generic relative phases between the two intermediate paths --- more
precisely, whenever $\mathrm{Re}[\,a(f_1;g_1) \cdot (a(f_2;g_2))^*\,] \neq
0$ --- this produces a nonzero cross-term, so the category is outside
$\Ccoh$ for that choice of weights. (A nontrivial relative phase alone
does not suffice: at relative phase exactly $\pi/2$, for instance, the
cross-term vanishes despite the phase being nontrivial, so the correct
condition is on the real part of the interference term, not on the phase
being nonzero per se.) A category $(C, w)$ lies in $\Ccoh$ only if interference
cross-terms through every intermediate object vanish identically --- a
strong structural condition, not a regularity one.
\end{example}

\section{Discussion}\label{sec:discussion}

\subsection{Application to categorical-refinement frameworks}

The motivating application of Theorem~\ref{thm:born} is to
categorical-refinement approaches to quantum mechanics in which physical
configurations are objects of a quiver-generated category, physical transformations
are morphisms, and quantum structure emerges from complex-valued morphism
weights rather than from prior Hilbert-space structure (T. Ustun,
``Categorical Refinement Dynamics and a Consistency Relation Among
Standard Model Flavor Parameters,'' Zenodo, DOI:
\href{https://doi.org/10.5281/zenodo.22069114}{10.5281/zenodo.22069114}).
In such frameworks, the Born rule has previously been invoked as an
informal consistency requirement. Theorem~\ref{thm:born} elevates this to
a proved local result under five explicitly stated axioms, separating what
can be settled locally from what requires a functorial probabilistic
structure.

\subsection{Scope and limitations}

The theorem characterizes the probability functional $P$ only locally, on a
single hom-set. It does not address: (a) the global coherence problem of
Section~\ref{sec:global}; (b) whether every complex-weighted category ---
not necessarily quiver-generated, a strictly broader class than the
theorem's own hypothesis --- actually admits amplitudes satisfying all
five axioms; in particular,
whether the classical-limit randomized-phase hypothesis of axiom~(iv) is
physically realized in a given framework; (c) the derivation of specific
state spaces, measurement operators, or operational structures. Each of
these requires input beyond the local axioms of Theorem~\ref{thm:born}.

\subsection{On the ``fixed-degree functional'' language, and extension to real-analytic \texorpdfstring{$P$}{P}}\label{sec:fixed-deg}

Informal statements of Born-rule uniqueness in the categorical-amplitude
literature frequently invoke ``fixed-degree functionals of $A$ and $\bar A$''
without precise specification. Theorem~\ref{thm:born} makes the notion
precise via axiom~(ii): a polynomial of bounded total degree $d$, with $d$
any finite positive integer. The proof shows that only the degree-$2$
component survives regardless of $d$, so the specific value is immaterial
provided it is finite.

The polynomiality axiom is strictly stronger than needed. If
$P(z) = \sum_{k=0}^\infty c_k\, |z|^{2k}$ is real-analytic and
$U(1)$-invariant, with the series $\sum_k c_k w^k$ converging absolutely
for $|w| < R$ (equivalently, $P$ converges absolutely on $|z| < \sqrt{R}$),
the proof of Theorem~\ref{thm:born} carries over on a restricted domain.
\emph{This restriction is necessary, not merely technical}: absolute
convergence of $\sum_k c_k w^k$ at radius $R$ does not imply convergence
at $2R$, and $|A_1 + e^{i\theta}A_2|^2$ can reach $(r_1+r_2)^2$ as $\theta$
varies --- so the exchange of sum and integral below is licensed only when
$(r_1+r_2)^2 < R$, i.e., $r_1 + r_2 < \sqrt{R}$. Under that restriction,
uniform convergence on the resulting compact subsets justifies exchanging the
infinite sum and the uniform-phase expectation in axiom~(iv), yielding
\[
\sum_{k \geq 0} c_k\, \mathbb{E}_\theta\bigl[\,|A_1 + e^{i\theta} A_2|^{2k}\,\bigr]
= \sum_{k \geq 0} c_k\, (r_1^{2k} + r_2^{2k})
\]
for $r_1, r_2 > 0$ with $r_1 + r_2 < \sqrt{R}$. The argument of Step~3 applies
at each order via the same homogeneous-degree decomposition: writing
$x = r_1^2$, $y = r_2^2$, each $\mathbb{E}_\theta[\,|A_1 + e^{i\theta}
A_2|^{2k}\,]$ is homogeneous of degree $k$ in $(x,y)$, so equality of the
two (now infinite) graded sums still forces equality degree by degree on
the tested region $r_1 + r_2 < \sqrt{R}$: $c_0 = 0$ at degree $0$, $c_1$
unconstrained at degree $1$ (normalization fixes $c_1 = 1$), and $c_k
\cdot k^2 = 0$ hence $c_k = 0$ at every degree $k \geq 2$, exactly as in
Step~3 --- no induction across degrees is needed here either, for the same
reason: distinct degrees never share a monomial. Since $r_1 + r_2 <
\sqrt{R}$ is still a full two-parameter continuum (not merely finitely
many test points), this pins down every coefficient $c_k$ of $P$'s
defining series as a global constant. The coefficients thus determined
--- $c_1 = 1$, all others $0$ --- apply to the \emph{original} series
$P(z) = \sum_k c_k|z|^{2k}$ wherever \emph{that} series converges, which is
all of $|z| < \sqrt{R}$: the restriction $r_1+r_2<\sqrt{R}$ was needed only
to justify the sum--integral interchange used to \emph{test} the
coefficients, not to bound where the resulting formula for $P$ holds. This
gives $P(z) = |z|^2$ throughout
the full convergence domain $|z| < \sqrt{R}$. Standard analytic continuation extends this to the
natural domain of $P$ --- i.e., the maximal open subset of $\C$ on which
$P$ admits a real-analytic extension; for $P$ globally real-analytic on
$\C$ this is all of $\C$, and since the conclusion $P(z) = |z|^2$ is itself
entire, the extension is trivially global in that case. If the natural
domain is a proper subset of $\C$, the uniqueness conclusion holds on that
subset only, and global extension requires the additional hypothesis of
global analyticity.

Hence the conclusion $P(z) = |z|^2$ holds for any real-analytic,
$U(1)$-invariant, non-negative $P$ satisfying axioms~(iv) and~(v) --- the
polynomiality axiom of Theorem~\ref{thm:born} is a convenience, not a
necessity. I retain polynomiality in the main theorem statement because it
makes the proof self-contained without invoking uniform convergence
arguments. Dropping analyticity further --- permitting merely continuous or
measurable $P$ --- would genuinely require a strengthening of axiom~(iv)
or a new continuity hypothesis, and is left open.

\subsection{Comparison with Zurek's envariance}

Zurek's derivation of the Born rule from environment-assisted
invariance~\cite{Zurek2003} relies on specific symmetries of entangled
bipartite states in Hilbert space. The present theorem does not assume a
bipartite structure or entanglement; it applies to arbitrary
complex-weighted quiver-generated categories and extracts the Born rule from a more
minimal symmetry --- the global $U(1)$ rephasing of a single amplitude ---
combined with classical-limit additivity. The two approaches are
complementary: envariance establishes $|\cdot|^2$ operationally in the
Hilbert-space setting using entanglement, while the present result
establishes it structurally in the categorical setting using phase symmetry
and a classical-limit consistency requirement.

\section*{Acknowledgments}

No external funding was received for this work. The author has no competing
interests to declare.


\end{document}